\documentclass[12pt]{article}

\usepackage{epsfig}

\usepackage{amssymb}
\usepackage{amsfonts}

\usepackage{color}
 
\def\be{\begin{equation}}
\def\ee{\end{equation}}
\def\ba{\begin{array}{c}}
\def\ea{\end{array}}

\def\ben{$$}
\def\een{$$}

\newcommand{\bea}{\begin{eqnarray}}
\newcommand{\eea}{\end{eqnarray}}

\newtheorem{thm}{Theorem}

\newtheorem{lemma}[thm]{Lemma}

\newenvironment{proof}{\noindent
 {\bf Proof.}}{\hfill$\square$\vspace{3mm}\endtrivlist}

\begin{document}

\begin{center}

{\Large

Unitarity corridors to exceptional points

}

\vspace{0.8cm}

  {\bf Miloslav Znojil}

\vspace{0.2cm}

\vspace{1mm} Nuclear Physics Institute of the CAS, Hlavn\'{\i} 130,
250 68 \v{R}e\v{z}, Czech Republic

{e-mail: znojil@ujf.cas.cz}

\end{center}

\section*{Abstract}

Phenomenological quantum Hamiltonians
$H^{(N)}(\lambda)=J^{(N)}+\lambda\,V^{(N)}(\lambda)$ representing a
general real $N^2-$parametric perturbation of an
exceptional-point-related unperturbed Jordan-block Hamiltonian
$J^{(N)}$ are considered. Tractable as non-Hermitian (in a
preselected, unphysical Hilbert space) as well as, simultaneously,
Hermitian  (in another, ``physical'' Hilbert space) these matrices
may represent a unitary, closed quantum system if and only if the
spectrum is real. At small $\lambda$ we show that the parameters are
then confined to a ``stability corridor'' ${\cal S}$ of the $\lambda
\to 0$ access to the extreme dynamical exceptional-point regime. The
corridors are $N-$dependent and narrow: They are formed by a
non-empty subset of unitarity-compatible multiscale perturbations
such that $\lambda\,V^{(N)}_{j+k,j}(\lambda) ={\cal
O}(\lambda^{(k+1)/2})\,$ at $k=1,2,\ldots,N-1\,$ and all $j$.

\subsection*{Keywords}

unitary quantum systems; perturbation theory; exceptional points;
admissible non-Hermitian Hamiltonians;
realizable perturbations;
boundaries of stability;

\newpage


\section{Introduction\label{troduction}}

One of the most characteristic distinguishing features of many
innovative {\em
non-Hermitian\,}  (e.g., ${\cal PT}-$symmetric \cite{Carl})
representations $H \neq H^\dagger$ of quantum
Hamiltonians is that they can vary with parameters which
are {\em allowed to reach\,} the Kato's exceptional-point values
(EPs, \cite{Kato}). The phenomenological appeal of such a limiting
transition $g\to g^{EP}$ in $H(g)$ is currently being discovered in
a broad range of open quantum systems \cite{Berry,Uwe,Nimrod} as
well as in many less known applications of the theory to various
closed quantum systems \cite{Dyson,Geyer,BG,BB,aliKG}. In the
former, open-system setting the spectrum of $H(g)$ is, in general, complex.
The $H(g)-$generated quantum time-evolution is
non-unitary. This gives rise to a number of rather unexpected and interesting
time-evolution patterns (for example, at $g= g^{EP}$ one could stop the
light \cite{stop}) which mainly attracted attention among
experimentalists \cite{Muslimani,Heiss,Cart,Carlbook}.

In the latter,
closed-system-oriented research,
in contrast,
the mainstream efforts are currently being concentrated upon
the study of many fundamental,
not yet fully resolved theoretical questions \cite{book}.
One of the most important ones concerns
the very relevance of the spectrum. Indeed,
under small perturbations, ``the location of the eigenvalues may be
\ldots fragile'' \cite{Trefethen} so that
people started believing that also
``in quantum mechanics
with non-Hermitian operators
\ldots a
central role'' is to be given to ``the mathematical concept
of the pseudospectrum'' \cite{Viola}.

Our present message is in fact mainly inspired by
the necessity of a critical comment on
the latter claims.
The point is that
the very definition of the ``smallness'' of perturbation $\lambda\,V$
only carries a well-defined physical
meaning
in
the mathematical descriptions of
non-unitary {\it alias\,} open quantum systems.
The claims of ``fragility'' are then firmly
based on the rigorous
Roch-Silberman theorem \cite{[64]}
``relating the pseudospectra to the stability of the spectrum under small
perturbations'' \cite{Viola}.
The use of pseudospectra related to the
perturbations with bounded norm $||V||={\cal O}(1)$ and with a small
coupling $\lambda < \epsilon$
then results, naturally, in the observation of
many ``unexpected wild properties of operators familiar
from ${\cal PT}-$symmetric quantum mechanics''
(cited, again, from \cite{Viola}).

All such claims are mathematically correct of course.
It is only necessary to add that
they exclusively apply to the open quantum systems.
In the case of closed quantum systems
the relationship between mathematics and physics is more subtle.
We are initially introducing our Hamiltonians
$H$ as non-Hermitian in
a conventional Hilbert space
(in our comprehensive review \cite{SIGMA}
we proposed to denote this space by
dedicated symbol ${\cal H}^{(F)}_{(friendly)}$).
In this space the norm $||V||$ and pseudospectra
are
defined \cite{Trefethen}.
Naturally,
as long as $H \neq H^\dagger$, such a space
has to be reclassified as auxiliary and manifestly
unphysical.
As a consequence,
it is necessary to
construct another, amended,
phenomenologically relevant
norm.
Only such a norm
can be used
in the formulations of
testable physical predictions concerning the
closed quantum systems
\cite{ali}.

In what follows we intend to
contribute to the clarification of the misunderstanding.
By means of a detailed analysis
of a few schematic examples
we intend
to
demonstrate that
one must treat the concept of a
`sufficiently small perturbation''
(entering also the definition of pseudospectrum)
with extreme care.
We will remind the readers that
in quantum mechanics
of unitary systems using observables in a non-Hermitian
representation
\cite{ali}
the weight of a perturbation is {\em not\,}
measured by its norm in
${\cal H}^{(F)}_{(friendly)}$.
By explicit constructive calculations we will clarify why
it must be measured by the norm
in another, {\em physical}, unitarily non-equivalent
Hilbert space of states with
standard probabilistic interpretation
(denoted, say,
by symbol
${\cal H}^{(S)}_{(standard)}$
of
Table 2 in~\cite{SIGMA}).

The difference between the two norms
increases when we get closer to the EP boundary of the
``admissible'' (i.e., unitarity-compatible) domain of parameters.
For this reason we
found it maximally instructive to
restrict attention of our readers
just to the systems
living in a small vicinity
of one of their EP singularities.
This enabled us to
make our message compact and persuasive.
Indeed,
whenever the system moves closer to
its EP boundary,
the inner-product-related
anisotropy of
geometry of the associated
physical Hilbert space
${\cal H}^{(S)}_{(standard)}$
grows and approaches its non-Hermitian-degeneracy
supremum
(cf. \cite{lotoreichik}).

We will show that and how this induces a
``hierarchization'' of the
weights of the influence of the separate components of the
fluctuations of the separate matrix elements of the Hamiltonian.
Indeed, even if we keep calling these fluctuations ``perturbations'',
we must include also their anisotropy-dependence fully into account.
Due to our choice of not too complicated illustrative examples
we will be able to
simplify some technicalities significantly.
The presentation of our results will start, in
section \ref{seca}, by a concise
explanation of the situation in which
the vicinity of the EP singularity can be connected,
by a continuous change of the parameters, with the
bulk parametric domain of a
less anomalous dynamical regime of the system.

Conveniently,
the admissible, unitarity-compatible parametric domain near
an EP
will be called
``corridor''. By definition, the energies inside the corridor
will be required
real.
The concept of the corridor connecting a stable unitary dynamical
regime with its limiting EP boundary is given a more concrete form
in section \ref{secb}. We recall and extend there a few constructive
results of our preceding papers \cite{admissiblea,admissible}. We also
reconfirm there that under the quite common \cite{Trefethen}
but not sufficiently
restrictive assumption that the perturbations are uniformly
bounded,  the vicinity of generic EP-limiting $H$s
{\em does not contain any\,}
``broad'' corridors at all.

The apparent paradox is resolved in section \ref{secc} where we
introduce a concept of a ``narrow'' corridor for which the
``sufficiently small'' perturbations are newly defined via a certain
{\it ad hoc\,} redefinition of the space of variability of the
``admissible'' matrix elements of perturbation $V$. Explicit
formulae for the boundaries of the corridors are presented there at
the first few matrix dimensions $N$.
The subsequent more general and $N-$independent results
will be then presented in section
\ref{secd}. In a way based on an
extrapolation of the preceding $N-$dependent
observations to all $N$ we will formulate there our main result.

This will only explicitly reconfirm our {\it a priori\,} expectations
that in the non-Hermitian closed-system theories the basic phenomenological
concept of the ``smallness'' of the stability-compatible
perturbations $V$ must be specified in a far from trivial manner.
In our last, less technical discussion in section \ref{discussion} we
will finally complement this conclusion
by a few
comments on its consequences and interpretation.

\section{Unitarity corridors\label{seca}}

In a way reflecting the recent trends \cite{MZbook} we intend
to perform a deeper analysis of the mathematical guarantees
of the reality of the spectrum attributed, often, to the
spontaneously unbroken ${\cal PT}-$symmetry of $H$ \cite{Carl}, or
to the existence of a similarity between $H$ and a self-adjoint
operator \cite{Dyson,Geyer}. Such a project led us to the search for
correspondence of the underlying mathematics with the parallel
conceptual physical questions concerning, first of all, the
protection of a quantum system against the loss of its observability
under too strong a perturbation.

\subsection{The boundaries of observability}

In the literature devoted to the analyses of quantum stability one
mostly finds just various entirely routine descriptions which mainly
fall into two
subcategories. In the more common approach one simply assumes that
both the unperturbed and perturbed Hamiltonians are self-adjoint.
This, in essence, makes the problem trivial. Indeed, the reality of
the bound state energies remains ``robust''. One also does not need
to pay too much attention to the EP singular values $g^{(EP)}$ of
parameters because they are, by definition, out of consideration,
incompatible simply with the self-adjointness assumption
\cite{Kato}.

In the conventional Hermitian theories the influence of small
perturbations
is described by the
pseudospectrum
and it
remains fully under our control. In the
open-system theories the study of pseudospectra clarifies a number
of features of various realistic systems. {\it Pars pro
toto\,} we may name the study of perturbations of the
Bose-Hubbard $N$-by-$N$-matrix
forms of Hamiltonians $H^{(N)}(g)$  \cite{Uwe}. In this case the
non-Hermitian formalism of perturbation expansions helped to
clarify even some aspects of the behavior of the Bose-Einstein
condensates. Another particularly impressive result of this type was
a quite unexpected discovery of the generic failure of adiabatic
approximation in the open, non-unitary quantum dynamical systems when forced to
encircle their EP singularity \cite{Doppler}.

The problems are much more challenging in the case of the closed
quantum systems, especially in the models in which the Kato's
EP singularity is of the $N-$th order with $N>2$ (in this case
we shall usually use the acronym EPN). Indeed, after an arbitrarily
small perturbation
the initially strictly non-diagonalizable EPN-related Hamiltonians
$H^{(N)}(g^{(EPN)})$ cannot be assigned their canonical Jordan-block
form anymore (i.e., they become diagonalizable). At the same time,
the brute-force numerical digonalization of these perturbed
Hamiltonians
 \be
 H^{(N)}=H^{(N)}(g^{(EPN)})+\lambda\,H_{(int)}^{(N)}\,
 \label{[2]}
 \ee
remains almost prohibitively ill-conditioned \cite{ngeq6}.
In what follows a remedy will be sought in perturbation theory
(cf. its outline in our preceding paper
\cite{admissible}). On this background
we will
separate perturbations
$H_{(int)}^{(N)}$ into two subfamilies.
For the subfamily of our present interest
(in which the energy spectra will be real) the parameters
will form a unitarity-compatible corridor.

Working, for the sake of definiteness,
with a multiparametric and real but, otherwise, entirely general $N$
by $N$ matrix Hamiltonians $H^{(N)}$ we shall restrict our study
just to the models lying ``not too far'' from an EP singularity. In
these models, secondly, the non-Hermitian EP degeneracy will be assumed
``maximal'', i.e., $N-$tuple, with $g^{EP}\,\equiv\,g^{EPN}$.
We should emphasize that
these restrictions of the scope of our paper were motivated by the
needs of physics. In particular, we wanted to complement the
open-system results of
Ref.~\cite{Uwe} or the
closed-system results of
Refs.~\cite{maximal}
(exhibiting  already all features of
a quantum phase transition \cite{Denis})
by another family of the less
realistic and less numerical
but
more universal and
more transparent $N$ by $N$ matrix model.

\subsection{Exceptional points and Jordan blocks}

In the EP limit itself (also known as non-Hermitian degeneracy
\cite{Berry}) the Hamiltonian, by definition, ceases to be
diagonalizable. This means that it loses its standard physical
interpretation \cite{Messiah}. At the same time, the study and
understanding of the behavior of quantum systems in the vicinity of
EPs is of paramount descriptive \cite{Uwe} as well as conceptual
\cite{Eva} and practical numerical \cite{ngeq6} relevance and
importance.

Special attention is to be paid to the scenarios
in which we may ignore the role of the EP-unrelated
part of the physical Hilbert space.
This enables us to restrict attention to the
$N$ by $N$ (sub)matrices
$H=H^{(N)}(g)$ of the Hamiltonian,
especially when the
parameter
is able to acquire a maximal,
$N-$th-order exceptional-point value, $g \to g^{(EPN)}$
\cite{maximal}.
In similar cases the EPN limit of the truncated
Hamiltonian is usually
assigned its Jordan-block canonical form,
 $$
 H^{(N)}(g^{(EPN)}) \sim J^{(N)}(E_0)=
 \left[ \begin {array}{ccccc}
                     E_0&1&0&\ldots&0
 \\\noalign{\medskip}0&E_0&1&\ddots&\vdots
 \\\noalign{\medskip}0&\ddots&\ddots&\ddots&0
 \\\noalign{\medskip}\vdots&\ddots&0&E_0&1
 \\\noalign{\medskip}0&\ldots&0&0&E_0
 \end {array} \right]\,.
 $$
A mutual map is defined, in terms of the so called transition matrix
$Q^{(N)}$, by relations
 \be
 H^{(N)}(g^{(EPN)})\, Q^{(N)} = Q^{(N)}\,J^{(N)}(E_0)\,.
 \label{ceases}
 \ee
In quantum physics, such a ``generalized diagonalization'' of the
Hamiltonian offers an efficient tool for analysis in perturbation
theory \cite{admissible}.

\subsection{One-parametric corridor in an exactly solvable example}

As an elementary illustrative example let us recall
the following exactly solvable $N-$state quantum
Hamiltonian of dimension $N=8$,
 \be
 H^{(8)}_{(ES)}(g)=\left[ \begin {array}{cccccccc} 0&-1+\delta&0&0&0&0&0&0
 \\\noalign{\medskip}-1-\delta&0&-1+{\it \gamma}&0&0&0&0&0
  \\\noalign{\medskip}0&-1-{\it \gamma}&0&-1+\beta&0&0&0&0
 \\\noalign{\medskip}0&0&-1-\beta&0&-1+\alpha&0&0&0
 \\\noalign{\medskip}0&0&0&-1-\alpha&0&-1+\beta&0&0
 \\\noalign{\medskip}0&0&0&0&-1-\beta&0&-1+{\it \gamma}&0
 \\\noalign{\medskip}0&0&0&0&0&-1-{\it \gamma}&0&-1+\delta
 \\\noalign{\medskip}0&0&0&0&0&0&-1-\delta&0
 \end {array} \right]\,.
 \label{eses}
 \ee
This matrix is non-Hermitian but ${\cal PT}-$symmetric, where ${\cal
P}$ is parity (i.e., a matrix with units along the second main
diagonal) while the nonlinear operator of transposition ${\cal T}$
mimics the time reversal \cite{BB}. In a standard
decomposition $H = T+V$ of this Hamiltonian the kinetic energy term
$T$ coincides with the conventional discrete Laplacean while the
four-parametric antisymmetric tridiagonal matrix $V$ plays the role
of a weakly non-local interaction.

The EP8 limit of the model is reached at
$\alpha=\beta=\gamma=\delta=1$. The resulting Hamiltonian matrix
with the mere $N-1$ non-vanishing elements $H_{j+1,j}=-2$,
$j=1,2,\ldots,N-1$ may be given its Jordan-block form via
Eq.~(\ref{ceases}) in terms of an antidiagonal transition matrix
with $N$  non-vanishing elements $Q_{N-j+1,j}=(-2)^{j-1}$,
$j=1,2,\ldots,N$.

For the specific $g-$dependence of parameters
 \be
 \alpha=\sqrt{1-a\,g}\,,\ \ \
 \beta=\sqrt{1-b\,g}\,,\ \ \
 \gamma=\sqrt{1-c\,g}\,,\ \ \
 \delta=\sqrt{1-d\,g}\,
 \ee
with a quadruplet of positive constants $ a, b, c$ and $d$ the
spectrum is sampled, in Fig.~\ref{trija}, at $a=2$, $b=1.8$, $c=1.6$
and $d=1.4$. It is all real and discrete at any real $g>0$. With
$g^{(EP8)}_{(ES)}=0$ and with the trivial degenerate energy
$E_0=E_{(ES)}^{(EP8)}=0$, the $g-$dependence of the energies can
even be specified by the remarkable exact formula
$E_n(g)=E_n(1)\,\sqrt{g}$ which immediately follows from the
$g-$dependence of the secular polynomial.

\begin{figure}[h]                    
\begin{center}                         
\epsfig{file=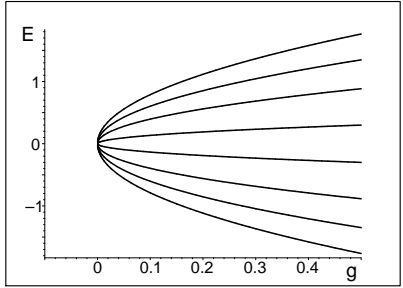,angle=270,width=0.30\textwidth}
\end{center}    
\vspace{2mm} \caption{The sample of degeneracy of the real spectrum
of Hamiltonian $H^{(8)}_{(ES)}(g)$ in the EP8 limit of $g \to 0$ at
$a=2$, $b=1.8$, $c=1.6$ and $d=1.4$ (both $g$ and $E$ are
dimensionless here).
 \label{trija}
 }
\end{figure}

\section{Exceptional points and their bounded perturbations\label{secb}}

The stability of quantum systems with respect to perturbations is
usually studied in the framework of conventional quantum mechanics
in which the Hamiltonians (i.e., the generators of evolution) are
self-adjoint \cite{Messiah}. From this perspective our present study
of manifestly non-Hermitian perturbed Hamiltonians (\ref{[2]})
living in a small vicinity of an EPN singularity represent a true
methodical challenge.

In the first step towards a disentanglement of the problems let us
recall Eq.~(\ref{ceases}) and let us replace the unperturbed
Hamiltonian $H^{(N)}(g^{(EPN)})$ of Eq.~(\ref{[2]}) by its canonical
Jordan form. This yields the fairly general family of the
EPN-related perturbed $N$ by $N$ real-matrix Hamiltonians of our
interest,
 \be
 H^{(N)}=J^{(N)}(0)+\lambda\,V\,.
 \label{illy}
 \ee
In our analysis we shall initially assume that the matrix elements
of the perturbation are uniformly bounded, $V_{i,j}={\cal O}(1)$.
The smallness of the perturbation then becomes controlled by a
single, ``sufficiently small'' positive parameter $\lambda$
tractable as a coupling constant.

\subsection{Exactly solvable $N=2$ model}


Jordan block with $N=2$,
 $$
 J^{(2)}(E_0)= \left[ \begin {array}{cc}
 E_0&1
 \\\noalign{\medskip}0&E_0\end {array}
 \right]
 $$
and with, say, $E_0=0$ can be perceived, in the light of
Eq.~(\ref{ceases}), as a generic representative of an arbitrary
$N=2$ one-parametric Hamiltonian $H^{(2)}(g)$ in its EP2 limit.
Thus, up to a trivial incorporation of transition matrix $Q^{(2)}$
via Eq.~(\ref{ceases}) we may replace any given  unperturbed
Hamiltonian $H^{(2)}(g^{(EP2)})$ in Eq.~(\ref{[2]}) by its canonical
form $J^{(2)}(0)$. Even when adding an arbitrary (and, say, real)
$N=2$ perturbation matrix
 $$
  V=\left[ \begin {array}{cc} \alpha_{{1}}&\mu
  \\\noalign{\medskip}\beta&\alpha_{{2}}\end {array} \right]
 $$
with bounded elements $V_{j,k}={\cal O}(1)$, the exhaustive
construction of all of the bound states remains non-numerical. Its
detailed presentation may be found in section III A. of
\cite{admissible}. For the present methodical
purposes we only need to recall the elementary scaling rule
 \be
 E_\pm^{(2)} =\pm \sqrt{\lambda\,\beta\,}+{\cal O}(\lambda)
 \label{sca2}
 \ee
characterizing the order of magnitude of the complete perturbed
bound-state energy spectrum. This rule immediately follows from
secular equation
 $$
 \det (H-E)=\left[ \begin {array}{cc} \lambda\,\alpha_{{1}}
 -\epsilon\,\sqrt {\lambda}&1+\lambda\,\mu
 \\\noalign{\medskip}\lambda\,\beta&\lambda\,
\alpha_{{2}}-\epsilon\,\sqrt {\lambda}\end {array} \right]=0\,,\ \ \
\ \epsilon=E /\sqrt {\lambda} ={\cal O}(1)
 $$
i.e., from the implicit definition of the spectrum
 $$
  \left( \alpha_{{1}}\alpha_{{2}}-\beta\,\mu \right) {\lambda}^{2}+
 \left( -\alpha_{{1}}\epsilon-\epsilon\,\alpha_{{2}} \right)
 {\lambda}^{3/2}+ \left( -\beta+{\epsilon}^{2} \right) \lambda=0\,.
 $$
The conclusion is that in the leading-order approximation we get the
two real energy roots $E_\pm$ of Eq.~(\ref{sca2}) if and only if
$\beta \geq 0$. In such a broad ``physical'' parametric corridor the
time-evolution of our quantum system remains unitary in a non-empty
interval of small $\lambda \in (0,\lambda_{\max})$. In contrast, the
eigenvalues become purely imaginary whenever $\beta<0$, $\epsilon
\approx \epsilon_\pm=\pm {\rm i}\sqrt{|\beta|}$. In other words, the
vicinity of the EP2 singularity splits into the  ``admissible'',
unitarity-compatible corridor and its ``unphysical'',
unitarity-incompatible complement. Thus, the choice of $\beta>0$
guarantees the existence of a non-empty corridor connecting the
interior of the domain of the stable dynamical regime with its
EP2-supporting boundary.

What remains to be discussed is the behavior of the $N=2$ bound
state energies in the limit $\beta \to 0$. Incidentally, for the
analysis the perturbation approximation approach is not needed. The
eigenvalue formulae $E_{1,2}=\lambda\,\alpha_{1,2}$ become exact at
$\beta=0$. What is new is only an enhancement of their order of
smallness, $E_{1,2}={\cal O}( {\lambda})$. We will see below that
such a rescaling behavior will also reemerge at the larger matrix
dimensions $N>2$.

%

\subsection{Nontrivial model with $N=3$}

The existence of transition matrices $Q^{(3)}$ and the routine
solvability of the EPN-related Eq. (\ref{ceases}) at $N=3$ enable us
to restrict attention, without any loss of generality, to the
perturbed Jordan-block Hamiltonians
 \be
 H^{(3)}(\lambda)=J^{(3)}(0)+\lambda\,V\,.
 \label{mod3}
 \ee
A partial analysis of consequences may already be found described in
section III C. of paper~\cite{admissible}. Unfortunately,
our conclusions in {\it loc. cit.} were negative. In the EP3
vicinity the quantum systems in question appeared non-unitary and
unstable. In the real space of parameters of perturbation $V$ we
failed to localize a unitarity-compatible corridor which would
provide a $\lambda \neq 0$ access to the EP3 singularity in the
limit of $\lambda \to 0$.

In retrospective, the main reason of the failure may be traced back
to the fact that we tried to follow the guidance provided by the
simpler $N=2$ model too closely. The use of the mere
$\lambda-$independent real perturbation matrix with elements
$V_{j,k}={\cal O}(1)$, i.e.,
 \be
  V=\left[ \begin {array}{ccc} \alpha_{{1}}&\mu_{{1}}&\nu
  \\\noalign{\medskip}\beta_{{1}}&\alpha_{{2}}&\mu_{{2}}
  \\\noalign{\medskip}\gamma&\beta_{{2}}&\alpha_{{3}}
 \end {array}
 \right]\,
 \label{pertu3}
 \ee
appeared insufficient. In fact, we only too heavily relied upon the
existence of the specific ``exact'' representation of the $N=3$
spectrum it terms of Cardano formulae.
After all, this strategy led already
to overcomplicated formulae and did not offer any
insight .
Thirdly, in a way guided by
the results at $N=2$ we ``skipped \ldots the discussion of
models with vanishing $\gamma = 0$''~\cite{admissible}.
In other words,
having
restricted our attention to the mere search for a ``broad'' corridor
with $\gamma \neq 0$ we missed the opportunity.
We did not manage to find {\em any\,} reasonable
construction of the corridor of stability at {\em any\,}
non-vanishing $\lambda \neq 0$ in Eq.~(\ref{mod3})
(see the list of the related comments at the end
of section Nr. III
in~\cite{admissible}).

The non-existence of the corridor at $N=3$ and $\gamma
\neq 0$ may be given an elementary proof.
In its outline let us return to ansatz (\ref{pertu3}).
We may rescale the
energies, in a
way recommended in \cite{admissible}, whenever $\gamma\neq 0$,
$E_n=\epsilon_n\,\sqrt[3]{\lambda}$. An implicit definition
of the spectrum is then immediately provided by secular equation
 $$
\det \left[ \begin {array}{ccc}
\lambda\,\alpha_{{1}}-\epsilon\,\sqrt
[3]{\lambda}&1+\lambda\,\mu_{{1}}&\lambda\,\nu\\\noalign{\medskip}\lambda
\,\beta_{{1}}&\lambda\,\alpha_{{2}}-\epsilon\,\sqrt [3]{\lambda}&1+
\lambda\,\mu_{{2}}\\\noalign{\medskip}\lambda\,{\it
\gamma}&\lambda\,\beta _{{2}}&\lambda\,\alpha_{{3}}-\epsilon\,\sqrt
[3]{\lambda}\end {array}
 \right]=0\,.
 $$
Although the resulting
secular polynomial is too long for being printed,
its
leading-order part
is short and yields the final,
explicit closed-form result
 $$
 \epsilon \approx \epsilon_{1,2,3} = \sqrt[3]{\gamma}\,.
 $$
This reconfirms that the {\em whole} spectrum cannot be real (and
the system compatible with unitarity) unless $\gamma=0$.

\section{Construction of the corridors\label{secc}}

In~\cite{admissible} we did not study the case of vanishing
$\gamma=0$ because we found it overcomplicated. Now we shall accept
a different strategy, assuming that the limiting constraint
$\gamma=0$ is only valid in the leading-order approximation in
$\lambda$. In other words we will consider generalized, manifestly
$\lambda-$dependent versions of perturbations.

\subsection{The corridor and its boundaries
at $N=3$}

Transition to
manifestly $\lambda-$dependent real perturbation matrices
 \be
  V=\left[ \begin {array}{ccc} \alpha_{{1}}&\mu_{{1}}&\nu
  \\\noalign{\medskip}\beta_{{1}}&\alpha_{{2}}&\mu_{{2}}
  \\\noalign{\medskip}\gamma&\beta_{{2}}&\alpha_{{3}}
 \end {array}
 \right] + \sqrt{\lambda}\,V' + \ldots\,,
 \ \ \
  V'=\left[ \begin {array}{ccc} \alpha_{{1}}'
  &\mu_{{1}}'&\nu'
  \\\noalign{\medskip}\beta_{{1}}'&\alpha_{{2}}'&\mu_{{2}}'
  \\\noalign{\medskip}\gamma'&\beta_{{2}}'&\alpha_{{3}}'
 \end {array}
 \right]\,,\ \
 \ldots\,
 \label{ameper}
 \ee
is an enrichment of the representation of dynamics at $N=3$. It
immediately leads us to a very natural resolution of the puzzle. Let
us now outline its main technical ingredients. Firstly, at
$\gamma=0$ we have to change the energy scaling:
$E_n=\epsilon_n\,\sqrt{\lambda}$. From the resulting amended secular
equation
 $$
 \det
 \left[ \begin {array}{ccc}
 \lambda\,\alpha_{{1}}-\epsilon\,
 \sqrt {\lambda}&1+\lambda\,\mu_{{1}}&\lambda\,\nu
 \\\noalign{\medskip}\lambda
 \,\beta_{{1}}&\lambda\,\alpha_{{2}}-\epsilon\,\sqrt {\lambda}&1+
 \lambda\,\mu_{{2}}
 \\\noalign{\medskip}{\lambda}^{3/2}{\it \gamma}'&\lambda
 \,\beta_{{2}}&\lambda\,\alpha_{{3}}-\epsilon\,\sqrt {\lambda}
 \end {array} \right]=0
 $$
we are allowed to omit all of the higher-order corrections as
irrelevant. Preserving merely the ${\cal O}({\lambda}^{3/2})$
leading-order part of secular equation
 $$
 \gamma'+\left (\beta_{{1}}+\beta_{{2}}
 \right )\,\epsilon-{\epsilon}^{3}=0
 $$
we only need to reflect the role and influence of the new parameter
$\gamma'$. In a preparatory stage we may try to simplify the task
and to fix, tentatively, $\gamma'=0$. This would yield the two
sample roots $\epsilon_{\pm}=\pm \sqrt{\beta_1+\beta_2}$ which are
both real if and only if $\beta_1+\beta_2 \geq 0$. Thus, relation
$\beta_1+\beta_2 = 0$ seems to offer the first nontrivial
specification of the boundary of the corridor at $\gamma'=0$.
Unfortunately, the property of reality of the third energy root
(which, in the leading-order approximation, vanishes) remains
uncertain. Thus, we have to return to the full-fledged analysis of
the model at $\gamma' \neq 0$. Along these lines we abbreviated
$\beta_{{1}}+\beta_{{2}}=3\varrho^2$ and came to the following $N=3$
result.

\begin{lemma} \label{lejed}
For Hamiltonians (\ref{mod3}) with small $\lambda$ and arbitrary
real perturbations (\ref{ameper}) the energy spectra are real for
parameters inside an EP3-attached corridor such that $\gamma=0$ and
 $
 \gamma'\in (-\varrho^3 ,\varrho^3 )\,
 $.
\end{lemma}
\begin{proof}
The graph of the curve
$
y(\epsilon)=
 \gamma'+\left (\beta_{{1}}+\beta_{{2}}
 \right )\,\epsilon-{\epsilon}^{3}
 $
(with zeros equal to the energies) diverges to $\pm \infty$ at large
and positive/negative $\epsilon$, respectively. The
$\gamma'-$independent derivative $ y'(\epsilon)=
 \beta_{{1}}+\beta_{{2}}-3\,{\epsilon}^{2}
 $
has zeros $\epsilon_\pm=\pm \varrho$
which determine the local minimum/maximum of $y(\epsilon)
$. It must be negative/positive, respectively, but
this is guaranteed by our constraint upon $\gamma'$.
\end{proof}


\subsection{Boundaries at $N=4$}

The perturbed Jordan-block Hamiltonians
 \be
 H^{(4)}(\lambda)=J^{(4)}(0)+\lambda\,V\,
 \label{mod4}
 \ee
will be studied here with the following reduced,
ten parametric real
perturbation matrix
 \be
  V=\left[ \begin {array}{cccc} \mu_{{1}}&0&0&0
  \\\noalign{\medskip}\alpha_{{1}}&\mu_{{2}}&0&0
  \\\noalign{\medskip}\beta_{{1}}&\alpha_{{2}}&\mu_{
{3}}&0\\\noalign{\medskip}{\it
\gamma}&\beta_{{2}}&\alpha_{{3}}&\mu_{{4}}
\end {array} \right]\,.
\label{[9b]}
 \ee
Bound state energies  $E_n=\epsilon_n\,\sqrt[4]{\lambda}$ may now be
defined via roots of secular equation
 $$
 \det (H-E)=
 \det \left[ \begin {array}{cccc} \lambda\,\mu_{{1}}
 -\epsilon\,\sqrt [4]{\lambda}&1&0&0
 \\\noalign{\medskip}\lambda\,\alpha_{{1}}&\lambda\,\mu_{
{2}}-\epsilon\,\sqrt [4]{\lambda}&1&0\\\noalign{\medskip}\lambda\,
\beta_{{1}}&\lambda\,\alpha_{{2}}&\lambda\,\mu_{{3}}-\epsilon\,\sqrt
[ 4]{\lambda}&1\\\noalign{\medskip}\lambda\,{\it
\gamma}&\lambda\,\beta_{{2}
}&\lambda\,\alpha_{{3}}&\lambda\,\mu_{{4}}-\epsilon\,\sqrt
[4]{\lambda }\end {array} \right]=0\,.
 $$
In the leading-order approximation this yields the
entirely elementary
quadruplet of solutions
 $$
 \epsilon_n \approx \gamma^{1/4}\,.
 $$
At both signs of non-vanishing real $\gamma$ two of those roots are
purely imaginary so that at arbitrarily small $\gamma\neq 0$ and
$\lambda\neq 0$ the perturbed system becomes non-unitary. In other
words, our quantum system with $\gamma\neq 0$ is unstable and it
does not possess any suitable physical Hilbert space of states ${\cal
H}^{(S})_{(standard)}$. The system must be interpreted as having performed a
phase transition at $\lambda=0$ \cite{BB,Denis}.
At $\gamma\neq 0$ and $\lambda\neq
0$ its ``energy'' $H$ is not an observable anymore.

The essence of the paradox was clarified in our preceding paper
\cite{admissible}. We emphasized there that in the quantum mechanics
of closed systems it only makes sense to consider the ``realizable''
perturbations under which the perturbed Hamiltonian still operates
in a suitable ${\cal H}^{(physical)}$. One has to require that the
``strength'' of the perturbations is ``measured'' in  ${\cal
H}^{(physical)}$ rather than in any of its unitarily non-equivalent,
manifestly unphysical alternatives  ${\cal H}^{(auxiliary)}$ with,
typically, a ``friendlier'' inner product \cite{ali}.

In our present continuation of the EP4-related study of realizable
perturbations let us reopen the search for a stable
corridor in a restricted parametric domain where $\gamma \approx 0$.
With this aim we replace the constant-perturbation ansatz
(\ref{[9b]}) by a more sophisticated, $\lambda-$dependent $N=4$
analogue of Eq.~(\ref{ameper}). Recalling the strategy used at $N=3$
we have to modify also the scaling of the bound state energies and
put $E_n=\epsilon_n\,\sqrt[3]{\lambda}$ in an amended secular
equation
 $$
 \det (H-E)=\det
 \left[ \begin {array}{cccc}
 \lambda\,\mu_{{1}}-\epsilon\,\sqrt [3]{\lambda}&1&0&0
 \\\noalign{\medskip}\lambda\,\alpha_{{1}}&\lambda\,\mu_{
 {2}}-\epsilon\,\sqrt [3]{\lambda}&1&0
 \\\noalign{\medskip}\lambda\,
 \beta_{{1}}&\lambda\,\alpha_{{2}}&\lambda\,\mu_{{3}}-\epsilon\,\sqrt [
 3]{\lambda}&1
 \\\noalign{\medskip}{\lambda}^{4/3}{\it \gamma}'&\lambda\,
 \beta_{{2}}&\lambda\,\alpha_{{3}}&\lambda\,\mu_{{4}}-\epsilon\,\sqrt [
 3]{\lambda}\end {array} \right]=0\,.
 $$
Its leading-order component of order ${\lambda}^{4/3}$ must vanish,
 \be
  {\epsilon}^{4}-\beta_{{1}}\epsilon-\beta_{{2}}\epsilon
 -{\it \gamma}' =0\,.
 \label{sekw}
 \ee
Such an upgrade of secular polynomial has still strictly two or four
complex roots. The evolution of the system remains non-unitary and
unstable unless we set $\beta_{{1}}+\beta_{{2}}\to 0$ and
$\gamma'\to 0$ making all of the roots of the new leading-order
secular equation (\ref{sekw}) vanish as well.

The way out of the difficulty is found in the next-step lowering of
the order of magnitude of all of the coefficients in approximate
Eq.~(\ref{sekw}). In the language of physics this means that we have
to introduce certain {\it ad hoc\,} higher-order perturbations.
Thus, proceeding along the same lines as before we weaken the
dominant components of the perturbation, $\beta_1 \to
\beta_1'\,\sqrt[2]{\lambda}$, $\beta_2 \to
\beta_2'\,\sqrt[2]{\lambda}$ and $\gamma' \to \gamma''\,{\lambda}$.
This induces the change in the scaling of the energies,
$E_n=\epsilon_n\,\sqrt[2]{\lambda}$. The replacements lead to the
following ultimate amendment of Schr\"{o}dinger operator
 \be
 H-E=
 \left[ \begin {array}{cccc}
 \lambda\,\mu_{{1}}-\epsilon\,\sqrt {\lambda}&1&0&0
 \\\noalign{\medskip}\lambda\,\alpha_{{1}}&\lambda\,\mu_{
 {2}}-\epsilon\,\sqrt {\lambda}&1&0
 \\\noalign{\medskip}{\lambda}^{3/2}
 \beta_{{1}}'&\lambda\,\alpha_{{2}}&\lambda\,\mu_{{3}}-\epsilon\,\sqrt {
 \lambda}&1
 \\\noalign{\medskip}{\lambda}^{2}{\it \gamma}''&{\lambda}^{3/2}
 \beta_{{2}}'&\lambda\,\alpha_{{3}}&\lambda\,\mu_{{4}}-\epsilon\,\sqrt {
 \lambda}
 \end {array} \right]\,.
 \label{schrop}
 \ee
Up to the higher-order ${\cal O}({\lambda}^{5/2})$
corrections the exact secular equation $\det (H-E)=0$ degenerates to
the vanishing of the effective secular polynomial,
 \be
 z(\epsilon)=
  -\widetilde{\gamma}
  -\widetilde{\beta}\,\epsilon
  -\widetilde{\alpha}\,\epsilon^2
  +{\epsilon}^{4}=0\,
  \ee
where
 \be
   \widetilde{\gamma}=\gamma'' -\alpha_{1}\alpha_{3}\,,
 \ \ \ \ \
 \widetilde{\beta}=\beta_{1}' +\beta_{2}'\,,
 \ \ \ \ \
 \widetilde{\alpha}=\alpha_{1}+\alpha_{2}+\alpha_{3}\,.
 \label{ar}
 \ee
We arrive at our final answer.


\begin{lemma}
 \label{ledva}
For a sufficiently small $\lambda$ in the hierarchically perturbed
four by four Hamiltonian of Eq.~(\ref{schrop}) the energy spectrum
remains real inside a non-empty EP4-attached corridor of parameters.
\end{lemma}

\begin{proof}
The graph of the left-hand-side function $z(\epsilon)$ of
Eq.~(\ref{ar}) (with its four zeros equal to the energies) has its
three real extremes localized at the zeros of its
$\widetilde{\gamma}-$independent derivative $ z'(\epsilon)=
 -\widetilde{\beta}-2\,\widetilde{\alpha}
 \epsilon+4\,{\epsilon}^{3}
 $.
In the proof of Lemma \ref{lejed} we saw that the latter triplet of
zeros $\xi_{0,\pm}$ was real for $\widetilde{\alpha}=3\varrho^2>0$
and $\widetilde{\beta} \in (-\varrho^3,\varrho^3)$. Thus, up to the
parameters at the endpoints of these constraints the zeros
$\xi_{0,\pm}$ (i.e., the coordinates of the local extremes of
$z(\epsilon)$) are real and non-degenerate. Thus, the local maximum
of $z(\epsilon)$ is sharply larger than both of the local minima,
$z(\xi_{0})> \max z(\xi_{\pm}) $. As a consequence, the interval of
variability of our last free parameter $\widetilde{\gamma}$
guaranteeing that $z(\xi_{0})> 0$ while $\max z(\xi_{\pm})<0$ is
non-empty.
\end{proof}

\section{Corridors at arbitrary $N$\label{secd}}

The form of $H$ in Eq.~(\ref{schrop}) is instructive in revealing a
general hierarchy of relevance of the individual matrix elements of
$V$ under the natural phenomenological requirement of the
preservation of the unitarity of the evolution. The pattern can
tentatively be extrapolated to the higher matrix dimensions $N$
with, in particular, $E=\epsilon\,\sqrt[2]{\lambda}$ in the $N=5$
Schr\"{o}dinger operator
 $$
 H-E=\left[ \begin {array}{ccccc}
  \lambda\,\nu_{{1}}-\epsilon\,\sqrt {\lambda}&1&0&0&0
  \\\noalign{\medskip}\lambda\,\mu_{{1}}&\lambda\,\nu_{{
2}}-\epsilon\,\sqrt
{\lambda}&1&0&0\\\noalign{\medskip}{\lambda}^{3/2}
\alpha_{{1}}'&\lambda\,\mu_{{2}}&\lambda\,\nu_{{3}}
-\epsilon\,\sqrt{\lambda}&1&0
\\\noalign{\medskip}{\lambda}^{2}\beta_{{1}}''&{\lambda}^{3/
2}\alpha_{{2}}'&\lambda\,\mu_{{3}}&\lambda\,\nu_{{4}}-\epsilon\,\sqrt
{ \lambda}&1
\\\noalign{\medskip}{\lambda}^{5/2}{\it
\gamma}'''&{\lambda}^{2}
\beta_{{2}}''&{\lambda}^{3/2}\alpha_{{3}}'&\lambda\,\mu_{{4}}&\lambda\,
\nu_{{5}}-\epsilon\,\sqrt {\lambda}\end {array} \right]
 $$
etc (see also the illustrative explicit rederivation of such a form
of the corridor-compatible matrix in subsection \ref{pard} below).

\subsection{Extrapolation pattern}

We saw that at $N=2$, $N=3$ and $N=4$ the perturbation-expansion
construction of the energy spectrum near the Jordan-block extreme
$H^{(N)}(g^{(EPN)})$ was straightforward. The same technique can
equally well be applied at any larger matrix dimension $N$. Our
specific additional physical requirement of the reality of the
spectrum (i.e., of the unitarity of the time evolution of the
quantum systems in question) has been found to be satisfied inside a
specific non-empty domain which we called corridor to EPN. We also
saw that at $N=2$, $N=3$ and $N=4$ the corridor can be defined by
certain very specific choice of $\lambda-$dependent perturbations
$\lambda\,V(\lambda)$ in which the matrix elements are of {\em
unequal\,} orders of smallness. The pattern appeared amenable to a
rigorous extrapolation beyond $N=4$.

\begin{thm}
At any $N=2,3,\ldots$ and for all sufficiently small $\lambda>0$ the
reality of the bound-state spectrum of energies
$E_n=\epsilon_n\,\sqrt[2]{\lambda}$ with ${\epsilon}_n={\cal O}(1)$
can be guaranteed by an appropriate choice of parameters
${\mu}_{jk}={\cal O}(1)$ in the real $N$ by $N$ matrix Hamiltonian
$H=J^{(N)}(0)+\lambda\,V$ with
 \be
 \lambda\,V=\left[ \begin {array}{cccccc}
  0&0&\ldots&0&0&0
  \\\noalign{\medskip}\lambda\,{\mu}_{{21}}&0&\ldots&0&0&0
  \\\noalign{\medskip}{\lambda}^{3/2}
  \,{\mu}_{{31}}&\lambda\,{\mu}_{{32}}&\ddots&\vdots&\vdots&0
  \\\noalign{\medskip}{\lambda}^{2}{\mu}_{{41}}&{\lambda}^{3/2}
  \,{\mu}_{{42}}
  &\ddots&0&0&0
  \\\noalign{\medskip}\vdots&\vdots&\ddots&\lambda\,{\mu}_{{N-1N-2}}&0&0
  \\\noalign{\medskip}{\lambda}^{N/2}{{\mu}}_{N1}&
 {\lambda}^{(N-1)/2}{\mu}_{{N2}}&\ldots&\lambda^{3/2}
 \,{\mu}_{{NN-2}}&\lambda\,{\mu}_{{NN-1}}&0
 \end {array} \right]\,.
 \label{uho}
 \ee
 \label{theofil}
\end{thm}
\begin{proof}
Once we guessed the appropriate $\lambda-$dependence of the general
Schr\"{o}dinger operator it is entirely straightforward to deduce
the general leading-order part of the secular determinant, and to
recall the independence and the free variability of the coefficients
in the secular polynomial.
\end{proof}

\subsection{The boundaries of corridor at $N=5$\label{pard}}



Let us start from the naive ten-parametric constant-matrix
perturbation
 $$
  V=\left[ \begin {array}{ccccc} 0&0&0&0&0\\\noalign{\medskip}\mu_{{1}}&0&0&0&0
  \\\noalign{\medskip}\alpha_{{1}}&\mu_{{2}}&0&0&0
\\\noalign{\medskip}\beta_{{1}}&\alpha_{{2}}&\mu_{{3}}&0&0
\\\noalign{\medskip}{\it \gamma}&\beta_{{2}}&\alpha_{{3}}&\mu_{{4}}&0
\end {array} \right]
 $$
and from the unperturbed Jordan-block matrix
 $
 H_0=J^{(5)}(0)
 $.
Schr\"{o}dinger operator with $E=\epsilon\,\sqrt[5]{\lambda}$ then
reads
 $$
 H-E=\left[ \begin {array}{ccccc} -\epsilon\,\sqrt [5]{\lambda}&1&0&0&0
 \\\noalign{\medskip}\lambda\,\mu_{{1}}&-\epsilon\,\sqrt [5]{\lambda}&1
&0&0\\\noalign{\medskip}\lambda\,\alpha_{{1}}&\lambda\,\mu_{{2}}&-
\epsilon\,\sqrt
[5]{\lambda}&1&0\\\noalign{\medskip}\lambda\,\beta_{{1
}}&\lambda\,\alpha_{{2}}&\lambda\,\mu_{{3}}&-\epsilon\,\sqrt [5]{
\lambda}&1\\\noalign{\medskip}\lambda\,{\it
\gamma}&\lambda\,\beta_{{2}}&
\lambda\,\alpha_{{3}}&\lambda\,\mu_{{4}}&-\epsilon\,\sqrt
[5]{\lambda}
\end {array} \right]\,.
 $$
After we reduce the secular polynomial to its dominant part we get
the five elementary energy roots
 $
 \epsilon \approx \gamma^{1/5}
 $.
Such a spectrum cannot be all real unless $\gamma=0$. This confirms
the necessity of diminishing the matrix element of perturbation in
its left lower corner, $\gamma \to \gamma'\,\sqrt[4]{\lambda}$ (we
may and will drop the primes). This forces us to change,
consistently, the scale of $E=\epsilon\,\sqrt[4]{\lambda}$. The
resulting new effective (i.e., leading-order) secular equation
 $
 \left( -{\epsilon}^{5}+\beta_{{1}}\epsilon
 +{\it \gamma}+\epsilon\,\beta_{{2}} \right){\lambda}^{5/4}=0
 $
is now found to lead, in nontrivial case, to at least two complex,
non-real energy roots. In the same corner of perturbation matrix as
above we have to diminish, therefore, the relevant matrix elements
again. Once we do so and once we drop the primes in $\beta_j \to
\beta_l'\,\sqrt[3]{\lambda}$, $\gamma' \to
\gamma''\,\sqrt[3]{\lambda^2}$ and $E=\epsilon'\,\sqrt[3]{\lambda}$
we get the following tentative amendment of Schr\"{o}dinger operator
 $$
 H-E=\left[ \begin {array}{ccccc} -\epsilon\,\sqrt [3]{\lambda}&1&0&0&0
 \\\noalign{\medskip}\lambda\,\mu_{{1}}&-\epsilon\,\sqrt [3]{\lambda}&1
&0&0\\\noalign{\medskip}\lambda\,\alpha_{{1}}&\lambda\,\mu_{{2}}&-
\epsilon\,\sqrt [3]{\lambda}&1&0\\\noalign{\medskip}{\lambda}^{4/3}
\beta_{{1}}&\lambda\,\alpha_{{2}}&\lambda\,\mu_{{3}}&-\epsilon\,\sqrt
[3]{\lambda}&1\\\noalign{\medskip}{\lambda}^{5/3}{\it
\gamma}&{\lambda}^{4
/3}\beta_{{2}}&\lambda\,\alpha_{{3}}&\lambda\,\mu_{{4}}&-\epsilon\,
\sqrt [3]{\lambda}\end {array} \right]\,.
 $$
Recycling the abbreviations of Eq.~(\ref{ar}) the dominant part of
the new effective secular equation acquires the explicit
three-parametric form
 \be
  -\widehat{\gamma}
  -\widetilde{\beta}\,\epsilon
  -\widetilde{\alpha}\,\epsilon^2
  +{\epsilon}^{5}=0\,,
  \ \ \ \
 \widehat{{\gamma}}=\gamma''\,.
 \label{arr}
 \ee
Its roots still cannot be all real unless they vanish in the given
order of precision. Making now the story short and iterating the
procedure once more we arrive, at last, at the ultimate hierarchized
and corridor-supporting perturbation matrix as given by Theorem
\ref{theofil},
 $$
 V=\left[ \begin {array}{ccccc} 0&0&0&0&0
 \\\noalign{\medskip}\mu_{{1}}&0&0&0&0
 \\\noalign{\medskip}\sqrt {\lambda}\,\alpha_{{1}}&
 \mu_{{2}}&0&0&0\\\noalign{\medskip}\lambda\,
 \beta_{{1}}&\sqrt {\lambda}\,\alpha_{{2}}&\mu_{{3}}&0&0
 \\\noalign{\medskip}{\lambda}\,\sqrt {\lambda}\,{\it \gamma}&
 \lambda\,\beta_{{2}}&\sqrt {\lambda}\,\alpha_{{3}}&
 \mu_{{4}}&0
 \end {array} \right]\,.
 $$
The effective $O({\lambda}^{5/2})$ part of the secular determinant
$\det (H-E)$ leads now to the explicit form of the polynomial
secular equation
 \ben
 -{\epsilon}^{5}+ \left( \mu_{{2}}+\mu_{{1}}+\mu_{{4}}+\mu_{{3}}
 \right) {\epsilon}^{3}+ \left( \alpha_{{1}}+\alpha_{{2}}+\alpha_{{3}}
 \right) {\epsilon}^{2}+
 \een
 \be
 +\left( \beta_{{1}}-\mu_{{2}}\mu_{{4}}-\mu_{{1
}}\mu_{{3}}+\beta_{{2}}-\mu_{{1}}\mu_{{4}} \right)
\epsilon-\alpha_{{1 }}\mu_{{4}}+{\it
\gamma}-\mu_{{1}}\alpha_{{3}}=0\,.
 \label{eqfive}
 \ee
It has the ultimate four-parametric flexibility as required. The
non-empty unitarity-preserving corridor to the $\lambda=0$ EP5
vertex does exist, with the leading-order boundaries prescribed, in
implicit but still user-friendly manner, by Eq.~(\ref{eqfive}).

\section{Discussion\label{discussion}}

The recent successful localizations of
the EP singularities
in various experimental setups
revealed
a perceivable increase of
their relevance
in applied physics as well as in the
quantum physics of resonant and unstable open systems \cite{Nimrod}.
In the quantum theory of stable systems
the role of EP singularities used to be traditionally restricted
to their purely mathematical role
of an obstruction of convergence
in perturbation theory \cite{Kato}.
Such a situation was only slowly improving with
the
emergence of the
first realistic models in relativistic quantum mechanics
where the EP marks an onset of instability \cite{aliKG}.
An analogous phenomenological phase-transition
interpretation was then also assigned to the EPs in many
other unitary quantum systems
\cite{BB,Denis,4a5}.

In a conventional perspective
these innovations
seem to contradict the well
known Stone theorem \cite{Stone}.
Due to this theorem any unitary evolution
(say, in
${\cal H}^{(S)}_{(standard)}$)
must
necessarily be generated by a
Hamiltonian which is selfadjoint
(naturally, in the same Hilbert space
${\cal H}^{(S)}_{(standard)}$).
From this point of view
the innovation of quantum theory
of unitary systems may be presented and advocated
in two ways.
Firstly, in an abstract manner,
as a purely
technical simplification of the physical inner product,
i.e., as a reduction of our standard physical Hilbert space
into its auxiliary partner,
i.e., as a replacement
 ${\cal H}^{(S)}_{(standard)}\to
{\cal H}^{(F)}_{(friendly)}$
leading to a
friendlier mathematics.
Secondly, alternatively,
in an opposite direction and in a very concrete spirit,
one picks up an auxiliary Hilbert space
first of all, Then one
replaces its unphysical but user-friendly
inner product
by a less friendly but correct physical amendment.

This is the most common formulation of the recipe.
In practice,
what is then required is just
the Hamiltonian-dependent construction
of the Hamiltonian-Hermitizing metric operator $\Theta$.
Naturally,  the existence of such a metric
requires the reality of the spectrum; there is no
consistent (unitary) quantum theory without such a constraint
\cite{Geyer}.
Once the spectrum is shown real,
we map ${\cal H}^{(F)}_{(friendly)}
\to {\cal H}^{(S)}_{(standard)}$
and convert
the initial, ``friendly but false'' Hilbert space with ``natural''
metric $\Theta^{(false)}=I$
into its model-dependent physical amendment with
metric $\Theta^{(standard)} \neq I$.

In the present continuation of the related considerations
in Ref.~\cite{admissible} we were able to explain
that the conjecture of the non-existence of an
``admissible'' access corridor to the EP3 limit only meant the
non-existence of a ``broad'' corridor (which we found to exist at
$N=2$ but not at $N=3$). We came now with a corrigendum:
The corridors of a
stable access to the EPN extremes do exist. The only
constraint is that
they are ``narrow'' in the sense of Theorem \ref{theofil}.

The reason of the
non-existence of a ``broad'' corridor at $N\geq 3$
has been shown here to lie in the fact that
at least some of the elements of
the class of perturbations which are only required
bounded in the auxiliary Hilbert space
${\cal H}^{(F)}_{(friendly)}$ may happen to be
too large in
${\cal H}^{(S)}_{(standard)}$.
Then, they can
move the system out of a given
(or, better, out of any eligible)
physical Hilbert space of course.
For this reason, the  perturbations
which are merely bounded
in the auxiliary space
${\cal H}^{(F)}_{(friendly)}$
become a purely formal construct because
they are only small with respect to
a phenomenologically irrelevant
metric $\Theta^{(false)}=I$.
Even without an explicit reference to the metric
we have shown that
at any dimension $N$ and at any, {\em arbitrarily small\,}
but nonvanishing $V^{(N)}_{N,1}={\cal O}(1)$ and $\lambda>0$
the perturbed Hamiltonians
(\ref{illy})
{\em cannot\,} be assigned {\em any\,}
physical meaning or experimental realization.

In the second, main step of our considerations we inverted the
ordering of questions.
In the light of our main interest in the system's stability
we decided to search
for a consistent, ``admissible'' subset of perturbations
$\lambda\,V$ which would still keep the perturbed Hamiltonian
compatible with the quantum  theory of reviews \cite{Carl,ali}.
We felt encouraged by a preparatory analysis of our
one-parametric illustrative model (\ref{eses})
which appeared easily converted into its canonical Jordan-block form.
Having
used these blocks as certain strong-coupling
EP-related unperturbed Hamiltonians
we were then able to leave the elementary model and to
extend the scope
of our considerations to the entirely general $N$ by $N$
real-matrix class
of perturbations $\lambda\,V^{(N)}$.

We may summarize that
we managed to
specify
the structure of admissible,
observability non-violating perturbation matrices
$V^{(N)} =V^{(N)}(\lambda)$ at all $N$.
Besides
the proof of existence
we also described
the method of an (implicit)
determination
of the leading-order boundaries of the
unitarity-compatible corridors ${\cal S}$
in the Euclidean real space of the
variable matrix elements of $V^{(N)}(\lambda)$.
Inside
these
domains of ``admissible'' parameters
the evolution
remains unitary.
We may conclude that in a way contradicting the scepticism
of conclusions based on the
constructions of the
pseudospectra \cite{Trefethen,Viola}
and/or of the ``broad'' corridors \cite{admissible},
the quantum systems in question
remain stable and closed inside corridors ${\cal S}$
which may be called ``narrow''.


\end{document}